\documentclass[twoside]{aiml22}

\usepackage{aiml22macro}

\usepackage{graphicx}
\usepackage{amsmath}
\usepackage{amssymb}

\newcommand{\bisim}{\;\underline{\leftrightarrow}\;}

\newcommand{\mono}{\mathcal{L}^+_{\Box,\Diamond,\wedge,\vee}}
\newcommand{\posexcon}{\mathcal{L}^+_{\Diamond,\wedge}}

\newcommand{\pos}{\mathcal{L}^+_{\Box,\Diamond,\wedge,\vee,\top,\bot}}
\newcommand{\negmono}{\mathcal{L}^-_{\Box,\Diamond,\wedge,\vee}}

\newcommand{\emptyloop}{\circlearrowright_{\emptyset}}
\newcommand{\fullloop}{\circlearrowright_{\mathrm{Prop}}}

\def\lastname{ten Cate, Koudijs}

\begin{document}

\begin{frontmatter}
  \title{Characterising Modal Formulas with Examples}
  \author{Balder ten Cate}\footnote{Supported by the European Union’s Horizon 2020 research and innovation programme under grant MSCA-101031081.}
  \address{ILLC University of Amsterdam}
  \author{Raoul Koudijs}\footnote{Special thanks to Nick Bezhanishvili's for his insights and helpful comments}
  \address{ILLC University of Amsterdam}

  \begin{abstract}
  We initiate the study of finite characterizations and exact learnability of modal languages. A finite characterization of a modal formula w.r.t. a set of formulas is a finite set of finite models (labelled either positive or negative) which distinguishes this formula from every other formula from that set. A modal language $\mathcal{L}$ admits finite characterisations if every $\mathcal{L}$-formula  has a finite characterization w.r.t.~$\mathcal{L}$.~This definition can be applied not only to the basic modal logic $\mathbf{K}$, but to arbitrary normal modal logics.
  We show that a normal modal logic admits finite characterisations (for the full modal language) iff it is locally tabular. This shows that finite characterizations with respect to the full modal language are rare, and hence motivates the study of finite characterizations for fragments of the full modal language.
  Our main result is that the positive modal language without the truth-constants $\top$ and $\bot$ admits finite characterisations. Moreover, we show that this result is essentially optimal: finite characterizations no longer exist when the language is extended with the truth constant $\bot$ or with all but very limited forms of negation.
  \end{abstract}

 \end{frontmatter}

\section{Introduction}
Every modal formula defines a possibly infinite set of pointed models that satisfy it (\textit{positive} examples), and implicitly also the set of pointed models that do not satisfy it (\textit{negative} examples). We study the question whether for a given class of modal formulas it is possible to \textit{characterize} every formula with a \textit{finite} set of positive and negative examples such that that no other formula is consistent with these. We call such sets of examples \textit{finite characterisations}. The existence of such finite characterizations is a precondition for the existence \emph{exact learning algorithms} for `reverse-engineering' a hidden goal formula from examples in Angluin's model of exact learning with membership queries \cite{AngluinQueriesConcepts}. Our interest in exact learnability is itself motivated by applications in description logic. But besides learnability, the generation of a finite exhaustive set of data examples consistent with a given logical specification, can be useful for illustration, interactive specification, and debugging purposes (e.g.,~\cite{MannilaR86} for relational database queries, \cite{AlexeCKT2011} for schema mappings, and \cite{StaworkoW15} for XML queries). The exhaustive nature of the examples is useful in these settings, as they essentially display all `ways' in which the specification can be satisfied or falsified.

In this extended abstract, we only provide a high level description of our results
and proof techniques. Detailed proofs can be found in~\cite{RaoulThesis}.

\section{Preliminaries}
Given a set of propositional variables $\mathrm{Prop}$ and a set of connectives $C\subseteq\{\wedge,\vee,\Diamond,\Box,\top,\bot\}$, let $\mathcal{L}_C[\mathrm{Prop}]$ (or simply $\mathcal{L}_C$ when $\mathrm{Prop}$ is clear from context) denote the collection of all modal formulas generated from literals (i.e. positive or negated propositional variables) from $\mathrm{Prop}$, using the connectives in $C$. Note that all such formulas are in negation normal form, i.e. negations may only occur in front of propositional variables. Further,
for any modal fragment $\mathcal{L}$ as defined above,
$\mathcal{L}^+$ and $\mathcal{L}^-$ denote the set of positive, respectively negative $\mathcal{L}$ formulas, where
a formula $\varphi$ is \textit{positive} if no $p\in var(\varphi)$ occurs negated, and \textit{negative} if all $p\in var(\varphi)$ occur only negated.
We will use \textit{modal language} to refer to any such fragment. 
By the \emph{full modal language} we will mean $\mathcal{L}_{\Box,\Diamond,\wedge,\vee,\top,\bot}[\mathrm{Prop}]$.

For a modal formula $\varphi$, let $var(\varphi)$ denote the set of variables occurring in $\varphi$ and $d(\varphi)$ its \textit{modal depth}, i.e. the nesting depth of $\Diamond$'s and $\Box$'s in $\varphi$.

A \textit{normal modal logic} is a collection of modal formulas containing all instances of the $K$-axiom $\Box(\varphi\to\psi)\to(\Box\varphi\to\Box\psi)$ and closed under uniform substitution, modus ponens and generalisation.
A (Kripke) \textit{model} is a triple $M=(dom(M),R,v)$ where $dom(M)$ is the a set of `possible worlds', $R\subseteq dom(M)\times dom(M)$ a binary `accessibility' relation and a valuation $V:\mathrm{Prop}\to\mathcal{P}(W)$. A \textit{pointed model} is a pair $M,s$ of a Kripke model $M$ together with a state $s\in dom(M)$. A (Kripke) \textit{frame} is a model without its valuation. 
\looseness=-1

\section{Finite Characterizations}
First, we give a formal definition of finite characterizations of modal formulas.

\begin{definition}{(Finite characterizations)} 
\label{def:characterization}
A \textit{finite characterization} of a formula $\varphi\in\mathcal{L}[\mathrm{Prop}]$ w.r.t. $\mathcal{L}[\mathrm{Prop}]$ is a pair of finite sets of finite pointed models $\mathbb{E}=(E^+,E^-)$ such that (i) $\varphi$ \textit{fits} $(E^+,E^-)$, i.e. $E,e\models\varphi$ for all $(E,e)\in E^+$ and $E,e\not\models\varphi$ for all $(E,e)\in E^-$ and (ii) $\varphi$ is the only formula in $\mathcal{L}[\mathrm{Prop}]$ which fits $(E^+,E^-)$, i.e. if $\psi\in\mathcal{L}[\mathrm{Prop}]$ satisfies condition (i) then $\varphi\equiv\psi$. A modal language $\mathcal{L}$ is \textit{finitely characterizable} if for every finite set of propositional variables $\mathrm{Prop}$, every $\varphi\in\mathcal{L}[\mathrm{Prop}]$ with has a finite characterization w.r.t. $\mathcal{L}[\mathrm{Prop}]$.
\end{definition}

Thus if $(E^+,E^-)$ is a finite characterization of a formula $\varphi\in\mathcal{L}[\mathrm{Prop}]$ w.r.t. $\mathcal{L}[\mathrm{Prop}]$, then for every $\psi\in\mathcal{L}[\mathrm{Prop}]$ with $\varphi\not\equiv\psi$, $E^+$ contains a finite model of $\varphi\wedge\neg\psi$ or $E^-$ contains a finite model of $\neg\varphi\wedge\psi$. For example, the formula $p\land q$ has a finite characterization w.r.t. 
$\mathcal{L}^+_{\wedge}[\mathrm{Prop}]$ with $\mathrm{Prop}=\{p,q,r\}$,
namely
$(\{\cdot_{p,q}\},\{\cdot_{p},\cdot_{q}\})$, where ``$\cdot_P$'' is the single point model where all $p\in P$ are true.

Our motivation for studying finite characterizations, comes from 
\emph{computational learning theory}. Specifically, finite characterizability is a  necessary precondition for \emph{exact learnability with membership queries} in Dana Angluin's interactive model of exact learning \cite{AngluinQueriesConcepts}. 
In our context, exact learnability with membership corresponds to a setting in which
the learner has to identify a formula by asking question to an oracle, where
each question is of the form ``is the formula true or false in pointed model $(M,w)$?''
This can also be viewed as a `reverse engineering' task, where a formula has to be identified based on its behaviour on only a finite set of models.
Exact learnability has recently gained a renewed interest in the description logic literature. We comment more on the connection with description logic in Section~\ref{sec:discussion}.

Our starting observation is:

\begin{theorem}\label{thm:forcingheight}
The full modal language is not finitely characterizable.
\end{theorem}
\begin{proof}
It suffices to give one counterexample, so suppose that e.g. $\varphi=\Box\bot$ had a finite characterization $(E^+,E^-)$ w.r.t.~the full modal language. 
Observe that for each $n$, $M,s\models\Box^{n+1}\bot\wedge\Diamond^n\top$ iff $height(M,s)=n$, where the \emph{height} of a pointed model $M,s$ is the length of the longest path in $M$ starting at $s$, or $\infty$ if there is no finite upper bound.
Every finite characterizations can only contain pointed models up to some bounded height, or one must have height $\infty$. In either case, for large enough $n$, 
$\varphi=(\Box^{n+1}\bot\wedge\Diamond^n\top)\vee\Box\bot$ is falsified on all negative examples in $E^-$, as $\Box\bot$ is also falsified on all these examples by hypothesis. Moreover, $\varphi$ is also true on all positive examples in $E^+$ since it contains $\Box\bot$ as a disjunct. However, clearly $\varphi\not\equiv\Box\bot$, contradicting our assumption that $(E^+, E^-)$ is a finite characterisation of $\Box\bot$.
\end{proof}

In fact, by a variation of the same argument, we can show that \textit{no} modal formula has a finite characterization w.r.t. the full modal language. 

Theorem~\ref{thm:forcingheight} raises two questions, namely: \emph{do finite
characterizations exist in other modal logics than $\mathbf{K}$}, 
and \emph{which fragments of modal logic admit finite characterizations}.
We address each of these two questions next.

We first generalize Definition~\ref{def:characterization} as follows (whereby Theorem \ref{thm:forcingheight} becomes a result about the special case of the basic normal modal logic $\mathbf{K}$): a finite characterization of a modal formula $\varphi$ with $var(\varphi)\subseteq\mathrm{Prop}$ w.r.t.~a normal modal logic $L$ is a finite set $(E^+,E^-)$ of finite pointed models based on $L$ frames such that (i) $\varphi$ fits $(E^+,E^-)$ and (ii) if $\psi$ with $var(\psi)\subseteq\mathrm{Prop}$ fits $(E^+,E^-)$ then $\varphi\equiv_L\psi$, where $\varphi\equiv_L\psi$ iff $\varphi\leftrightarrow\psi\in L$. We say that a normal modal logic $L$ is finitely characterizable if for every finite set $\mathrm{Prop}$, every modal $\varphi$ with $var(\varphi)\subseteq\mathrm{Prop}$ has a finite characterization w.r.t. $L$. We can give a complete characterization over which modally definable frame classes the full modal language is finitely characterizable.

It turns out that only very few normal modal logics are uniquely characterizable.
A normal modal logic $L$ is \textit{locally tabular} if for every finite set $\mathrm{Prop}$ of propositional variables, there are only finitely many formulas $\varphi$ with $var(\varphi)\subseteq\mathrm{Prop}$ up to $L$-equivalence. 

\begin{theorem}
A normal modal logic $L$ is finitely characterizable iff it is locally tabular.
\end{theorem}

In other words the full modal language is only finitely characterizable in the degenerate case where there are only finitely many formulas to distinguish from (up to equivalence). This result motivates the investigation of finite characterizability for modal fragments. 
Specifically, inspired by previous work on finite characterizability of the positive existential fragment of first order logic \cite{dalmautencateCQ}, we consider positive fragments of the full modal language. 

In the remainder of this section, we only consider again the modal logic $\mathbf{K}$. The proof of Theorem~\ref{thm:forcingheight} can easily be modified to show the following:%
\footnote{It suffices to replace $\top$ by a fresh propositional variable $q$ in the proof of Theorem~\ref{thm:forcingheight}.}

\begin{theorem}\label{thm:forcingheightforposminustop}
$\mathcal{L}^+_{\Box,\Diamond,\wedge,\vee,\bot}$ is not finitely characterizable.
\end{theorem}

On the other hand, based on results in \cite{dalmautencateCQ}, we can show that:

\begin{theorem}[From~\cite{dalmautencateCQ}]\label{thm:posexcon}
$\posexcon$ is finitely characterizable. Indeed, given a formula in $\posexcon$,
we can construct a finite characterization in polynomial time. 
\end{theorem}

More precisely, it was shown in \cite{dalmautencateCQ} (building on results in~\cite{FoniokNT08}) that finite characterizations can be constructed in polynomial time for ``c-acyclic conjunctive queries'', a fragment of first-order logic that includes the standard translations of $\posexcon$-formulas. 

Our main result here extends Theorem~\ref{thm:posexcon} by showing that $\mono$ is finitely characterizable.

\begin{theorem}\label{thm:main}
$\mono$ is finitely characterizable.
\end{theorem}

Theorem~\ref{thm:forcingheightforposminustop} above shows that this is essentially optimal; we leave open the question whether the fragment without $\bot$ but with $\top$ is finitely characterizable. 

In the rest of this section, we outline the ideas behind the proof
of Theorem~\ref{thm:main}. A key ingredient is the novel notion of \emph{weak simulation}, which we obtain by weakening the back and forth clauses of the \emph{simulations} studied in \cite{KurtoninaDeRijke}. Simulations are themselves a  weakening of bisimulations. It was shown in \cite{KurtoninaDeRijke} that $\pos$ is characterized by preservation under simulations.\looseness=-1

A \textit{weak simulation} between two pointed models $(M,s),(M',s')$ is a relation $Z\subseteq M\times M'$ such that for all $(t,t')\in Z$:
\begin{align*}
&(\mathrm{atom})         &&M,s\models p\quad\mathrm{implies}\quad M',s'\models p\\
&(\mathrm{forth}')               &&\mathrm{If}\;R^{M}tu,\;\mathrm{either}\;M,u\bisim\emptyloop\\[-1mm]&&&\text{or there is a}\;u'\;\mathrm{with}\;R^{M'}t'u'\;\mathrm{and}\;(u,u')\in Z\\
&(\mathrm{back}')               &&\mathrm{If}\;R^{M'}t'u',\;\mathrm{either}\;M',u'\bisim\fullloop\\[-1mm]&&&\text{or there is a}\;u\;\mathrm{with}\;R^{M}tu\;\mathrm{and}\;(u,u')\in Z
\end{align*}
where $\emptyloop$ denotes the single reflexive point with empty valuation, $\fullloop$ denotes the single reflexive point with full valuation and $\bisim$ denotes bisimulation. If such $Z$ exists, we say that $M',s'$ \emph{weakly simulates} $M,s$. The crucial weakening is witnessed by the fact that the deadlock model, i.e. the single point with no successors, weakly simulates $\emptyloop$, but does not simulate it.

Because weak simulations are closed under relational composition, which is associative, the collection of pointed models and weak simulations forms a category with $\emptyloop$ and $\fullloop$ as weak initial and final objects, respectively. 

\begin{theorem}\label{thm:preservation}
$\mono$ is preserved under weak simulations.
\end{theorem}

In high level terms, the proof of Theorem~\ref{thm:main} proceeds as follows:
given a formula $\varphi\in\mono$, 
we show how to construct positive and negative examples $(E^+_{\varphi}, E^-_{\varphi})$ that $\varphi$ fits and which forms a \textit{duality} (a generalisation of the notion of \textit{splittings} in lattice theory \cite{McKenzie1972EquationalBA}) in the in the category of pointed models and weak simulations. More specifically, we show that every model of $\varphi$ weakly simulates some positive example in $E^+$ and that every non-model of $\varphi$ is weakly simulated by some negative example in $E^-$. 
It follows by Theorem~\ref{thm:preservation} that any $\mono$-formula that fits $E^+$ is implied by $\varphi$, while every formula that fits $E^-$ implies $\varphi$, showing that $(E^+_{\varphi},E^-_{\varphi})$ is a finite characterization of $\varphi$ w.r.t. $\mono$.

This proof technique was inspired by results in~\cite{AlexeCKT2011}, which 
established a similar connection between finite characterizations for 
GAV schema mappings (or, equivalently, unions of conjunctive queries) and homomorphism dualities (i.e. generalised splittings in the category of finite structures and homomorphisms. 

\section{Discussion}\label{sec:discussion}
Our construction, although effective, is non-elementary. For this reason, we cannot obtain from it an efficient exact learning algorithm.
On the other hand, it follows from the results in~\cite{dalmautencateCQ} that
$\posexcon$-formulas are polynomial-time exactly learnable with membership queries.
We leave it as future work to prove matching lower bounds for our construction, 
and to understand more
precisely which modal fragments admit polynomial-sized finite characterizations 
and/or are polynomial-time exactly learnable with membership queries.

Variants of Theorem~\ref{thm:main} can be obtained for $\negmono$ and, more generally,
for \emph{uniform} modal formulas, where certain propositional variables
only occur positive and others only negatively.

As we mentioned, our immediate motivation for this work came from a renewed interest in 
exact learnability in description logic. In particular, in~\cite{dalmautencateCQ}, exact learnability with membership queries is studied for the description logic $\mathcal{ELI}$. These results are extended to results on learning $\mathcal{ELI}$ concepts under DL-Lite ontologies (i.e. background theory) \cite{Funk2021ELr} and temporal instance queries formulated in fragments of linear-time logic (LTL) \cite{LTLWolter}. 
We expect that our proof of Theorem~\ref{thm:main} can be lifted to the poly-modal case without major changes. 

\bibliographystyle{aiml22}
\bibliography{aiml22}

\end{document}